\documentclass[10pt,letterpaper,twoside,journal]{IEEEtran}

\usepackage{amsmath, amssymb}
\usepackage{mathrsfs}
\usepackage{cite}
\usepackage{graphicx}
\usepackage{algorithmic}
\usepackage{array}
\usepackage{textcomp}
\usepackage{multirow}
\usepackage{booktabs}
\usepackage{bm}
\usepackage{balance}
\usepackage{xcolor}       
\usepackage{todonotes}    

\usepackage[most]{tcolorbox} 

\newcommand{\Gh}{\textbf{G}}
\newcommand{\ds}{\textbf{d}}
\newcommand{\sg}{\sigma}

\newcommand{\sz}{_{0, \sigma}}
\newcommand{\E}{\mathrm{e}}

\usepackage{url} 

\usepackage{epstopdf}

\newcommand{\R}{\mathbb{R}}

\usepackage{graphics} 
\usepackage{epsfig} 
\usepackage{mathptmx} 
\usepackage{times} 
\usepackage{amsmath} 
\usepackage{amssymb}  
\newtheorem{theorem}{Theorem}
\newtheorem{lemma}{Lemma}
\newtheorem{corollary}{Corollary}

\newtheorem{remark}{Remark}
\newtheorem{definition}{Definition}

\makeatother

\begin{document}

\title{
 Accelerated Stabilization of Switched Linear MIMO Systems using Generalized Homogeneity
}

\author{Moussa Labbadi, Andrey Polyakov and Denis Efimov
\thanks{M. Labbadi is with the IRDL, UMR CNRS 6027, Bretagne INP, ENSTA, Institut
Polytechnique de Paris, Univ. Brest, Univ. Bretagne Sud, Brest, France (e-mail: \textsl{moussa.labbadi@enib.fr}). }
\thanks{A. Polyakov and D. Efimov are with the Inria, Univ. Lille, CNRS, UMR 9189 - CRIStAL, F-59000 Lille, France (e-mail: \textsl{denis.efimov@inria.fr}, \textsl{andrey.polyakov@inria.fr}).}
}

\maketitle
\IEEEoverridecommandlockouts

\thispagestyle{empty} 

\begin{abstract}
This paper addresses the problem of exponential and accelerated finite-time, as well as nearly fixed-time, stabilization of switched linear MIMO systems. The proposed approach relies on a generalized homogenization framework for switched linear systems and employs implicit Lyapunov functions for control design, covering both common and multiple Lyapunov function settings. Linear matrix equations and inequalities are derived to characterize the dilation generator and to synthesize the controller gains. Robustness of the resulting control laws with respect to system uncertainties and external disturbances is analyzed. The effectiveness of the proposed approach is illustrated through numerical examples.
\end{abstract}
\begin{IEEEkeywords}
Switched linear MIMO systems;
Homogeneous control;
Finite-time and fixed-time stability;
Robust stabilization; Generalized Homogeneity.
\end{IEEEkeywords}

\section{Introduction}
Hybrid systems are characterized by the interaction of continuous and discrete dynamic behaviors and have garnered substantial interest due to their importance in both theoretical research and practical applications~\cite{Goebel2012}. Switched systems are a subclass of hybrid dynamics, consisting of multiple modes and a switching signal that governs transitions among these modes~\cite{liberzon2003,Goebel2012}. In recent years, rapid advances in intelligent control have motivated extensive research on stability analysis and controller synthesis for switched systems, see~\cite{Lin2009,Long2013,yang2014survey,liberzon2003,ishii2002,shorten2007,zhao2012}, among others. 

A primary tool in this domain is the common Lyapunov function approach, while more recently, techniques based on multiple Lyapunov functions have been widely applied. Notably, \cite{liberzon2003,branicky1998} investigate the stability of switched linear systems using Lyapunov functions. In the review article \cite{shorten2007}, the authors provide an in-depth tutorial on stability analysis of switched linear systems and linear differential inclusions, presenting conditions for the subsystems involved in these models. The concept of employing multiple Lyapunov functions to assess the stability of switched systems was first introduced in \cite{branicky1998}, marking a significant development in the field. This idea has since stimulated extensive research on stability analysis using multiple Lyapunov functions. Additionally, \cite{zhao2012} relaxes the traditionally strict non-increasing condition on the derivative of Lyapunov functions, considering nonnegative linear systems.

The rate at which estimation or regulation errors converge to zero is a key performance metric in many real-world applications. Further studies on finite-time stability and robust control synthesis for uncertain switched systems are presented in \cite{garg2015, hu2021finite,Orlov2004}, where both necessary and sufficient conditions are established to guarantee stability under arbitrary switching scenarios. Extensions to time-delay systems are discussed in \cite{zhang2022fixed}. Additionally, sufficient criteria for finite-time stability of impulsive switched systems have been derived based on common~\cite{kumar2023finite} and multiple~\cite{zhang2022finite} Lyapunov functions, combined with a dwell-time condition.

Recently, significant attention has been devoted to applying the theory of homogeneous systems for control synthesis in both SISO (Single-Input, Single-Output) and MIMO (Multiple-Input, Multiple-Output) systems~\cite{Polyakov2016, Polyakov2025, Zimenko2020a, zimenk2018generalized, Polyakov2015a, Polyakov2012}. The core idea is to design controllers or observers that render the closed-loop or error system homogeneous. This approach is commonly referred to as homogenization~\cite{Zimenko2020a}. In such systems, the degree of homogeneity plays a critical role in determining convergence rates. Specifically, the local stability of a homogeneous system is closely linked to its global stability: if an asymptotically stable system has a negative degree of homogeneity, it is finite-time stable~\cite{Orlov2004}; conversely, a positive degree implies nearly fixed-time stability~\cite{Polyakov2025}.

Early work on finite-time stability of homogeneous switched systems was carried out in \cite{Orlov2004}, where it is shown that for switched systems with negative homogeneity degree, global asymptotic stability implies global finite-time stability. In \cite{Nersesov2008}, finite-time stabilization of nonlinear impulsive systems is established using Lyapunov‑based conditions. Finite-time stability of hybrid dynamical systems is further analyzed in \cite{Chen2012}, where Lyapunov criteria guarantee convergence to the origin in finite time. Global finite-time stabilization of switched nonlinear systems using control laws with positive odd rational powers was developed in \cite{Fu2015}. Cyclic switched nonlinear systems are shown to achieve finite-time stability under conditions derived in \cite{Yang2015}. Finite-time stability results with explicit settling-time estimates for nonlinear impulsive systems are provided in \cite{Li2019}.



Motivated by the above discussion, this paper studies feedback homogenization and robust stabilization of switched linear MIMO systems. The main contributions are summarized as follows. First, it is shown that switched linear MIMO systems can be rendered homogeneous of a prescribed degree by means of a linear state feedback combined with a homogeneous correction term. This property enables a unified treatment of exponential, finite-time, and nearly fixed-time stabilization within a generalized homogeneity framework. Second, both common and multiple homogeneous Lyapunov function approaches are developed. In the multiple Lyapunov function case, global uniform stability properties are established under an average/minimum dwell-time condition, explicitly characterizing the effect of mode-dependent Lyapunov functions and switching signals. Third, the controller synthesis is formulated in terms of linear matrix equations and inequalities, which characterize the generator matrix of the dilation and allow systematic tuning of the controller gains. Fourth, robustness with respect to uncertainties and external disturbances is analyzed. Sufficient conditions are derived to ensure that the finite-time, exponential, and nearly fixed-time stability properties are preserved in the presence of perturbations. Finally, numerical simulations are provided to validate the theoretical results and to illustrate the effectiveness of the proposed approach.

The remainder of the paper is organized as follows. Section~II formulates the stabilization problem for switched linear MIMO systems in the presence of perturbations. Section~III introduces the necessary preliminaries, including stability notions, generalized homogeneity, and implicit Lyapunov functions. Section~IV presents the main stabilization results and the proposed control laws. Section~V provides numerical
simulations, and Section~VI concludes the paper.
\subsection*{Notation}
\begin{itemize}
    \item \(\mathbb{R}\) denotes the field of real numbers, and \(\mathbb{R}^+ = \{ \alpha \in \mathbb{R} : \alpha \geq 0 \}\) represents the non-negative reals.
    
    \item \(\| x \| = \sqrt{x^\top x}\) is the Euclidean norm of \(x \in \mathbb{R}^n\).
    
    \item \(P \succ 0\) indicates that \(P  \in \mathbb{R}^{n \times n}\) is a symmetric and positive definite matrix. The relations \(\prec 0\), \(\succeq 0\), and \(\preceq 0\) refer to various definiteness conditions for matrices.
    
    \item \(\ds(s) = \E^{s \Gh_\ds}\), where \(s \in \mathbb{R}\), is a linear dilation in \(\mathbb{R}^n\) with anti-Hurwitz matrix \(\Gh_\ds\).
    
    \item \(\| x \|_P = \sqrt{x^\top P x}\) is the weighted Euclidean norm with \(P \succ 0\) and \(P \Gh_\ds + \Gh_\ds^\top P \succ 0\). 
    \item \(S = \{ x \in \mathbb{R}^n : \| x \|_P = 1 \}\) is the unit sphere in the weighted Euclidean space.
    
    \item \(\| \cdot \|_\ds\) represents the canonical homogeneous norm induced by \(\| \cdot \|\) \cite{Polyakov2025}.
    
    \item \(H_\ds(\mathbb{R}^n)\) is the set of all \(\ds\)-homogeneous scalar functions on \(\mathbb{R}^n\).
    
    \item \(F_\ds(\mathbb{R}^n)\) is the set of all \(\ds\)-homogeneous vector fields on \(\mathbb{R}^n\).
    
    \item \(\text{deg}_\ds(g)\) is the homogeneity degree of the function \(g \in H_\ds(\R^n)\) or \(g\in F_\ds(\R^n)\).
    \item  \( I_n \) is the identity matrix in \( \mathbb{R}^{n \times n} \).
    \item  $\Re(\lambda)$ denotes the real part of a complex number $\lambda$.
\end{itemize}

Preliminary results have been presented in our conference work~\cite{labbadi2025feedback}. Compared to~\cite{labbadi2025feedback}, this article provides additional details, including new definitions (Definition~\ref{def:ds-homogeneous}) and complete proofs. In particular, the analysis of multiple Lyapunov function settings is extended through Theorems~\ref{theorem:free_disturbances_MLF} and~\ref{thm:disturbances2}, while the case of global stability is addressed in Corollary~\ref{cor:sddt}.

\section{Problem Statement}
We consider the following switched linear plant:
\begin{equation} \label{eq:switched_linear}
\begin{aligned}
    \dot{x}(t) = A_{\sigma(t)} x(t) + B_{\sigma(t)} u(t) + E_{\sigma(t)}\omega(t, x(t)), \\ \sigma(t) \in \Sigma= \{1, \ldots, N\},
    \end{aligned}
\end{equation}
where $x(t)\in\mathbb{R}^n$ denotes the system state, $u(t)\in\mathbb{R}^m$ is the control input, and $A_{\sigma(t)} \in \mathbb{R}^{n \times n}$ and $B_{\sigma(t)} \in \mathbb{R}^{n \times m}$ are the system and input matrices associated with the active mode $\sigma(t)$. For each mode, the pair $\left(A_{\sigma(t)}, B_{\sigma(t)}\right)$ is assumed to be controllable. The matrix $E_{\sigma(t)} \in \mathbb{R}^{n \times p}$ characterizes the distribution of external disturbances, and $\omega(t, x(t))$ represents an external perturbation.

The state $x(t)$ and the piecewise-continuous switching signal $\sigma(t)$ are assumed to be fully measurable. To rigorously define system solutions in the presence of possible discontinuities in $\omega(t,x)$, Filippov’s differential inclusion framework is adopted~\cite{Filippov1960}.

The problem addressed in this paper is to design a constructive stabilizing feedback control law based on generalized homogenization techniques such that the switched system is robustly stabilized under arbitrary switching and  uncertainties. The proposed control design aims to ensure exponential, finite-time, or nearly fixed-time convergence by exploiting the homogeneity properties of the closed-loop system, while providing systematic tuning rules and robustness guarantees with respect to external disturbances.

\section{Preliminaries}

\subsection{Stability Notions}
%
We recall the  stability notions for an autonomous system:
\begin{align}\label{eq:SN}
\dot{x}(t) = f(x(t)),\; x(t)\in\R^n,\;t\geq0,		
\end{align}
where $f:\R^n\to\R^n$ ensures existence and uniqeness of solutions of the system. For brevity we assume that the system is forward complete, meaning that for  any \( x_0 \in \mathbb{R}^n \), there is a unique solution \( x(t; x_0) \).

\begin{definition}(see \cite{Bhat2000, Orlov2004})
The origin of \eqref{eq:SN} is said to be globally finite-time stable if it is globally asymptotically stable and, for any initial condition \( x_0 \in \R^n \), the solution $x(t; x_0)$ reaches the origin in finite time: \( x(t; x_0) = 0 \) for all \( t \geq T(x_0) \), where \( T(x_0) \) is called a settling-time function and satisfies \( T(x_0) < +\infty \) for all \( x_0 \in\R^n \).
\end{definition}
\begin{definition}(see \cite{Polyakov2012})
A subset \( M \subset \mathbb{R}^n \) is  globally finite-time attractive for system \eqref{eq:SN} if any trajectory \( x(t; x_0) \)  reaches \( M \) in a finite time \( t = T(x_0) \) and remains there for all \( t \geq T(x_0) \). Furthermore, the set \( M \) is said to be fixed-time attractive if the settling-time function is uniformly bounded by a constant, i.e., \( T(x_0) \leq T_{\max} \) for all \( x_0 \in\R^n\).
\end{definition}

\begin{theorem}(see \cite{Bhat2000,Laghrouche2017})
Let \( V :\R^n\to\R^+\) be a positive definite \( C^1 \) function on an open neighborhood of the origin \( \mathcal{D} \subseteq \mathbb{R}^n \), and there exist  real constants \( \kappa > 0 \) and \( \alpha \geq 0 \), such that for the system (\ref{eq:SN}):
\begin{equation}
    \dot{V}(x) \leq -\kappa V^{\alpha}(x), \quad x \in \mathcal{D}.
\end{equation}
Then, depending on \( \alpha \), the origin exhibits stability with the following types of convergence:
\begin{enumerate}
    \item If \( \alpha = 1 \), the origin is asymptotically  stable.
    \item If \( 0 \leq \alpha < 1 \), the origin is finite-time stable and settling time satisfies \( T(x_0) \leq \frac{1}{\kappa(1-\alpha)} V_0^{1-\alpha} \), where \( V_0 = V(x_0) \).
    \item If \( \alpha > 1 \), the origin is asymptotically stable and, for any \( \iota > 0 \), the set \( \mathcal{B} = \{x \in \mathcal{D} : V(x) < \iota\} \) is fixed-time attractive with \( T_{\max} = \frac{1}{\rho(\alpha-1)} \iota^{\alpha-1} \) (the origin is nearly fixed-time stable).
\end{enumerate}
If \( \mathcal{D} = \mathbb{R}^n \) and \( V \) is radially unbounded, the system exhibits these properties globally.
\end{theorem}
\subsection{Generalized Homogeneity}
\begin{definition}(see \cite{Polyakov2025,Zubov1958,Kawski1991})
A map \( \ds : \mathbb{R} \to \mathbb{R}^{n \times n} \) is a linear continuous  dilation in \( \mathbb{R}^n \) if:
\begin{enumerate}
    \item \( \ds(0) = I_n \), and \( \ds(t+s) = \ds(t)\ds(s) \) for all \( t, s \in \mathbb{R} \).
    \item \( \ds \) is continuous.
    \item \( \lim_{s \to -\infty} \|\ds(s)x\| = 0 \), and \( \lim_{s \to +\infty} \|\ds(s)x\| = +\infty \), for all \( x \in \mathbb{R}^n, \|x\| = 1 \).
\end{enumerate}
\end{definition}

\begin{theorem}(see \cite{Polyakov2025})
If \( \ds \) is a dilation in \( \mathbb{R}^n \), then the  generator matrix \( \Gh_\ds \) is anti-Hurwitz, and there exists a matrix \( P \in \mathbb{R}^{n \times n} \) such that
    \[
        P \Gh_\ds + \Gh_\ds^\top P \succ 0, \quad P \succ 0.
    \]
\end{theorem}
\begin{definition}\label{def:vector_fied} (see \cite{bhat2005geometric, Kawski1991}): A vector field \( \phi: \mathbb{R}^n \to \mathbb{R}^n \) (respectively, a function \( g: \mathbb{R}^n \to \mathbb{R} \)) is said to be \( \ds \)-homogeneous of degree \( \mu \in \mathbb{R} \) if  
\(
\phi(\ds(s)x) = \E^{\mu s} \ds(s) \phi(x)
\)
(respectively,  
\(
g(\ds(s)x) = \E^{\mu s} g(x)),
\)
for all \( x \in \mathbb{R}^n \) and \( s \in \mathbb{R} \). 
\end{definition} 

The following result holds:  
\begin{theorem}\label{thm:Orlov} (see \cite{Orlov2004,bhat2005geometric}) 
 Let $\phi : \mathbb{R}^n \to \mathbb{R}^n$ be a piecewise continuous vector field that is  $\ds$-homogeneous of degree $\mu \in\R$. If the system  $
    \dot{x} = \phi(x)$
is locally asymptotically stable at the origin, then it is globally  finite-time stable for $\mu<0$ and nearly fixed-time stable for $\mu>0$.
\end{theorem}
\subsection{On the Matrix Equation $XA - AX = X$}
Let us consider the folowing matrix equation. 
\begin{equation}
    XA - AX = X \quad \text{for} \quad  A,X\in\R^{n\times n}.
    \label{eq:matrix_equation}
\end{equation}  
\begin{lemma}(see \cite{zimenk2018generalized,Burde2005})\label{lem:Matrix}
Every matrix solution $X \in \mathbb{R}^{n \times n}$ of the equation \eqref{eq:matrix_equation} is nilpotent and hence satisfies $X^n = 0$.
\end{lemma}
Equation~\eqref{eq:matrix_equation} can be viewed as a particular case of the Sylvester matrix equation, whose general solution is available in~\cite{Gantmacher1998}. In the context of the inverse problem, that is, finding the matrix $A$ for a prescribed matrix $X$, the following result applies.
\begin{lemma}\label{lem:nilpotent_solution}\cite{zimenk2018generalized}
For any nilpotent matrix $X \in \mathbb{R}^{n \times n}$, equation~\eqref{eq:matrix_equation} has a solution.
\end{lemma}
\subsection{Stabilization of Linear MIMO Systems}
Consider a linear system:
\begin{equation} \label{eq:linear}
    \dot{x}(t) = A x(t) + B u(t),
\end{equation}
where $x(t)\in\mathbb{R}^n$ and $u(t)\in\mathbb{R}^m$ are the state and the control,    \( A \in \mathbb{R}^{n \times n} \) and \( B \in \mathbb{R}^{n \times m} \) are system and input matrices. The pair \( \{A, B\} \) is controllable.  The state vector \( x(t) \) is assumed to be fully measurable.
\begin{theorem}  \label{thm:2020} 
(Inspired by \cite{Polyakov2015a,Polyakov2025,Zimenko2020a,zimenk2018generalized})\\
If (\ref{eq:linear})  is \( \ds \)-homogeneously stabilizable with a degree \( \mu \neq 0 \), then it can be stabilizable for any degree. The control law can be defined,  in the form:
\[
u(x) = K_0 x + \|x\|_\ds^{\mu(1+\gamma)+\varepsilon} K \ds(-\ln \|x\|_\ds) x,
\]
where $\gamma\in\mathbb{R}$, \( K_0 \in\mathbb{R}^{m\times n}\), \( \varepsilon>0,\) and \( K \in\mathbb{R}^{m\times n} \) are determined by the following conditions:
\begin{itemize}
    \item \( K_0 = Y_0 (L-(\gamma+1) I_n)^{-1} \) for some \( L \in \mathbb{R}^{n \times n}, \, Y_0 \in \mathbb{R}^{m \times n}\), subject to:
    \[
    A L - L A - A + B Y_0 = 0, \ (L - \gamma I_n) B = 0, \ L - (\gamma + 1) I_n < 0.
    \]
    \item The matrix \( K =YX^{-1}  \) satisfies for $\delta,\eta>0$:
    \begin{gather*}
    (A + B K_0) X + X (A + B K_0)^\top + B Y + Y^\top B^\top + \delta X \leq 0, \\
    \quad X > 0, \quad \eta X \geq \mu L X + \mu X L^\top + 2 \varepsilon X > 0.
    \end{gather*}
\end{itemize}
\end{theorem}

\section{Main result}
We introduce the following definition of a $\ds$-homogeneous switched system.
\begin{definition}[$\ds$-homogeneous switched system]\label{def:ds-homogeneous}
Consider the switched system
\begin{equation}\label{eq:ds_switched}
\dot{x}(t)=f_{\sigma(t)}(x), \quad f_{\sigma} \in C(\mathbb{R}^n;\mathbb{R}^n), \quad x(0)=x_0,
\end{equation}
where \(\sigma(t)\) is a switching signal taking values in a finite index set.
Let \(x(t,x_0,\sigma)\) denote the solution of the switched system under an admissible switching signal \(\sigma\) with initial condition \(x_0\).
The system~\eqref{eq:ds_switched} is said to be $\ds$-homogeneous of degree \(\mu\) if, for all
\(s>0\), all initial conditions \(x_0\), all admissible switching signals
\(\sigma(\cdot)\), and all \(t \ge 0\), the corresponding solution satisfies
\begin{equation}\label{eq:ds_solutions_switched}
x\bigl(t,\, \ds(s)x_0,\, \sigma_s\bigr)
= \ds(s)\, x\bigl(\E^{\mu s} t,\, x_0,\, \sigma\bigr),
\end{equation}
where \(\ds(s)\) denotes the dilation operator and \(\sigma_s(t)=\sigma(\E^{\mu s} t)\)
is the time-rescaled switching signal.
\end{definition}

To formulate our main results we need the following auxiliary statements. 
\begin{theorem}
\label{thm:switched-homogeneous}
If each subsystem of the switched system~\eqref{eq:ds_switched} is \(\ds\)-homogeneous of degree \(\mu\) in the sense of Definition~\ref{def:vector_fied}, then the switched system is \(\ds\)-homogeneous of degree \(\mu\) in the sense of Definition~\ref{def:ds-homogeneous}.
\end{theorem}

\begin{IEEEproof}
Fix \(s>0\),  the time-rescaled switching signal $
\sigma_s(t)$ 
and the scaled trajectory
\[
y(t) := \ds(s)\, x(\E^{\mu s} t, x_0, \sigma), \quad t\ge 0.
\]

Since \(x(\cdot,x_0,\sigma)\) is absolutely continuous, \(y(t)\) is also absolutely continuous. By the chain rule, for almost all \(t\ge 0\),
\begin{gather*}
\dot y(t) = \E^{\mu s}\, \ds(s)\, \frac{\partial}{\partial\tau} x(\tau, x_0, \sigma) \Big|_{\tau = \E^{\mu s} t}
\\ = \E^{\mu s}\, \ds(s)\, f_{\sigma(\E^{\mu s} t)}\bigl(x(\E^{\mu s} t, x_0, \sigma)\bigr).
\end{gather*}

Using the \(\ds\)-homogeneity of \(f_{\sigma(\E^{\mu s} t)}\), we have
\[
\dot y(t) = f_{\sigma(\E^{\mu s} t)}\bigl(\ds(s) x(\E^{\mu s} t, x_0, \sigma)\bigr)
= f_{\sigma_s(t)}(y(t)).
\]

Moreover, the initial condition satisfies
$y(0) = \ds(s) x_0.$
Hence, \(y(t)\) is a solution of the switched system under \(\sigma_s\) with initial condition \(y(0)=\ds(s)x_0\). By the uniqueness of Carathéodory solutions \cite[Theorem 1.1, Chapter 1]{Filippov1988}, it follows that
\begin{equation}
y(t) = x(t, \ds(s)x_0, \sigma_s), \quad \forall t\ge 0.
\end{equation}

Substituting the definition of \(y(t)\), we obtain
\begin{equation}
x(t, \ds(s)x_0, \sigma_s) = \ds(s)\, x(\E^{\mu s} t, x_0, \sigma), \quad \forall t\ge 0,
\end{equation}
which shows that the switched system~\eqref{eq:ds_switched} is \(\ds\)-homogeneous of degree \(\mu\) according to Definition~\ref{def:ds-homogeneous}.
\end{IEEEproof}
Based on Theorem~\ref{thm:switched-homogeneous}, the following corollary provides the condition for a  switched linear  system \eqref{eq:switched_linear} to be \( \ds \)-homogeneous with the generator \( \Gh_\ds \) (i.e., each subsystem is \ds-homogeneous of the same degree).
\begin{corollary}\label{cor:linearswitching}
 Consider a  linear  switched system described by
\[
\dot{x}(t) = C_{\sigma(t)} x(t), \quad x \in \mathbb{R}^n, \quad C_{\sigma(t)} \in \mathbb{R}^{n \times n},
\]
where \( \sigma : \mathbb{R}_+ \to \Sigma \) is a piecewise constant switching signal. The system is \( \ds \)-homogeneous of degree \( \mu \in \mathbb{R} \) if and only if, for all \( i \in \Sigma \),
\begin{equation}\label{eq:commutator}
C_i \Gh_\ds - \Gh_\ds C_i = \mu C_i .
\end{equation}
\end{corollary}
\begin{IEEEproof}
\emph{Sufficiency.} Assume that, for each mode $i \in \Sigma$, the matrix $C_i$ satisfies~\eqref{eq:commutator} then, we obtain
\[
C_i \E^{\Gh_\ds s}
= \E^{(\Gh_\ds + \mu I_n)s} C_i,
\qquad \forall s \in \mathbb{R}.
\]
Hence, for each subsystem $\dot{x} = C_i x$, the identity
\[
C_i \ds(s) = \E^{\mu s} \ds(s) C_i
\]
holds for all $s \in \mathbb{R}$, which implies that each subsystem is
$\ds$-homogeneous of degree $\mu$. Since the homogeneity property is preserved under arbitrary switching, the switched system is $\ds$-homogeneous of degree $\mu$.

\emph{Necessity.}
Conversely, assume that the switched system is $\ds$-homogeneous of degree $\mu$.
Then, for each $i \in \Sigma$, the corresponding subsystem $\dot{x} = C_i x$
satisfies
\[
C_i \ds(s) = \E^{\mu s} \ds(s) C_i,
\qquad \forall s \in \mathbb{R}.
\]
Define
\[
M_i(s) := C_i \E^{\Gh_\ds s} - \E^{(\Gh_\ds + \mu I_n)s} C_i.
\]
By homogeneity, $M_i(s) = 0$ for all $s \in \mathbb{R}$. Differentiating with respect to $s$ yields
\[
\frac{d}{ds} M_i(s)
= C_i \Gh_\ds \E^{\Gh_\ds s}
- (\Gh_\ds + \mu I_n)\E^{(\Gh_\ds + \mu I_n)s} C_i = 0.
\]
Evaluating at $s = 0$ gives
\[
C_i \Gh_\ds - (\Gh_\ds + \mu I_n) C_i = 0,
\]
which is equivalent to~\eqref{eq:commutator}. This completes the proof.
\end{IEEEproof}

Based on Corollary \ref{cor:linearswitching}, it can be shown that the  switched system (\ref{eq:switched_linear}) with \( \omega(t,x) = 0 \) can be homogenized by a linear control \( u = K_{0,\sigma} x \).

\begin{lemma}\label{lemma:nilpotent}
If there exist matrices $K_{0,i} \in \mathbb{R}^{m \times n}$ and a matrix
$L_0 \in \mathbb{R}^{n \times n}$ such that the following conditions hold for all $i \in \Sigma$:
\begin{align}\label{eq:sylvester_conditions}
(A_i + B_i K_{0,i}) L_0 - L_0 (A_i + B_i K_{0,i}) = A_i + B_i K_{0,i},
\end{align}
then, for any $\mu \in \mathbb{R}$, there exists a matrix $\Gh_\ds$ such that the switched system~(\ref{eq:switched_linear})
with $\omega(t,x) = 0$ and $u = K_{0,\sigma} x$ is $\ds$-homogeneous of degree $\mu\neq 0$.
Conversely, if such a matrix $\Gh_\ds$ exists, then (\ref{eq:sylvester_conditions}) holds and each matrix $A_i + B_i K_{0,i}$ is nilpotent.
\end{lemma}
\begin{IEEEproof}
\emph{Sufficiency.}
Consider the closed-loop subsystem
\[
\dot{x} = (A_i + B_i K_{0,i}) x.
\]
According to Corollary~\ref{cor:linearswitching}, this subsystem is
$\ds$-homogeneous of degree $\mu$ if and only if
\begin{align}\label{eq:hom_cond_proof}
(A_i + B_i K_{0,i}) \Gh_\ds - \Gh_\ds (A_i + B_i K_{0,i}) = \mu (A_i + B_i K_{0,i}).
\end{align}
Define the generator matrix
\[
\Gh_\ds := \mu L_0 + I_n.
\]
Substituting such a $\Gh_\ds$ into~\eqref{eq:hom_cond_proof} and using~\eqref{eq:sylvester_conditions}, we obtain
\begin{gather*}
(A_i + B_i K_{0,i}) \Gh_\ds - \Gh_\ds (A_i + B_i K_{0,i}) 
= (A_i + B_i K_{0,i})(\mu L_0 + I_n) \\ - (\mu L_0 + I_n)(A_i + B_i K_{0,i}) 
= \mu \big[(A_i + B_i K_{0,i}) L_0 - L_0 (A_i + B_i K_{0,i})\big] \\
= \mu (A_i + B_i K_{0,i}),
\end{gather*}
which shows that condition~\eqref{eq:hom_cond_proof} holds for all $i \in \Sigma$.
Hence, the switched system is $\ds$-homogeneous of degree $\mu$ with generator $\Gh_\ds$.

\emph{Necessity.}
Assume that there exists a matrix $\Gh_\ds$ such that the switched system is $\ds$-homogeneous of degree $\mu$. Then, for each $i \in \Sigma$, condition~\eqref{eq:hom_cond_proof} holds. 
If $\mu \neq 0$, define $\bar{\Gh}_\ds := \frac{1}{\mu} \Gh_\ds$. Then
\[
(A_i + B_i K_{0,i}) \bar{\Gh}_\ds - \bar{\Gh}_\ds (A_i + B_i K_{0,i}) = A_i + B_i K_{0,i},
\]
which is exactly~\eqref{eq:sylvester_conditions}. By Lemma~\ref{lem:Matrix}, the existence of such a solution implies that $A_i + B_i K_{0,i}$ is nilpotent for all $i$.
\end{IEEEproof}

Further, these results are utilized for the control design of the switched  system (\ref{eq:switched_linear}).  First, the case of the common implicitly defined Lyapunov function is considered (as usual, in such a case the switching signal $\sigma(t)$ can be an arbitrary piecewise continuous function of time).
\begin{theorem}\label{theorem:free_disturbances}
Consider the switched linear  system \eqref{eq:switched_linear} with \(\omega(t,x) = 0\). Assume that
\begin{enumerate}
\item  The linear algebraic equation
\begin{align}\label{eq:homogeneous_eq0}
A_\sigma \Gh_0 - \Gh_0 A_\sigma + B_\sigma Y_{0,\sigma} &= A_\sigma, \quad \sigma \in\Sigma \\
\Gh_0 B_\sigma &= 0
\label{eq:homogeneous_eq}
\end{align}
has a solution \( Y\sz \in \mathbb{R}^{m \times n}, \Gh_0 \in \mathbb{R}^{n \times n} \), then for any solution of \eqref{eq:homogeneous_eq0} and \eqref{eq:homogeneous_eq}:
   \begin{itemize}
   
\item   The matrix
   \begin{equation}
   \Gh_\ds = I_n + \mu \Gh_0
   \label{eq:Gh_def}
   \end{equation}
   is anti-Hurwitz for $\mu\in \R$ being close enough to zero;
   \item The matrix \( \Gh_0 - I_n \) is invertible and the matrix
   \begin{equation}
   A\sz = A_\sg + B_\sg Y\sz (\Gh_0 - I_n)^{-1}
   \label{eq:A0_def}
   \end{equation}
   satisfies the following identities, for all $\sigma\in\Sigma$,
   \begin{equation}
   A\sz \Gh_\ds = (\Gh_\ds + \mu I_n) A\sz, \quad \Gh_\sg B_\sg = B_\sg;
   \label{eq:homogeneous_identities}
   \end{equation}
   \end{itemize} 
\item  The linear algebraic system
   \begin{align}\label{eq:homogeneous_identitiesf}
   \begin{aligned}
    X A_{0,\sigma}^{\top} + A_{0,\sigma} X + B_{\sigma} Y_\sigma + Y_\sigma^\top B_\sigma^\top \\ \quad 
    + \rho \left(\Gh_\ds X + X \Gh_\ds^{\top} \right) &\preceq 0, \;\forall\sigma\in\Sigma\\
     \Gh_\ds X + X \Gh_\ds^{\top} \succ0, \qquad 
     X \succ0,
     \end{aligned}
\end{align}
where \( K_{0,\sigma} =Y_{0,\sigma}(\Gh_0-I_n)^{-1}\)
has solution   \( X \in \mathbb{R}^{n \times n} \), \( Y_\sigma \in \mathbb{R}^{m \times n} \), for some \( \rho >0\). 
\end{enumerate}
 Then \begin{itemize}
\item the canonical homogeneous norm \( \| \cdot \|_{\ds} \) induced by
   \(
    \sqrt{x^\top P x}, \quad P =X^{-1}
   \) 
   is a (common) Lyapunov function for the closed-loop system \eqref{eq:switched_linear} with the feedback control
   \begin{equation}\label{eq:u}
    u(x) = K_{0,\sigma} x + \|x\|_\ds^{1+\mu} K_\sigma \ds(-\ln \|x\|)x,  
\end{equation}
where \( \ds(s) = \E^{s \Gh_\ds}, s \in \mathbb{R} \) is a linear dilation generated by \( \Gh_\ds \) and
   \begin{equation}
    K_{\sg} = Y_\sigma X^{-1}.
   \label{eq:K_matrices}
   \end{equation}
   Moreover,
   \begin{equation}
   \frac{d}{dt} \|x\|_{\ds} \leq -\rho \|x\|_{\ds}^{1+\mu}, \quad \forall x \neq 0.
   \label{eq:lyapunov_decay}
   \end{equation}

\item the control law \( u(x) \) satisfies
   \begin{equation}
   |u(x)| \leq |K\sz x| + \tilde{k}_\sigma \|x\|_{\ds}^{1+\mu}, \quad \forall x \in \mathbb{R}^n \setminus \{0\},
   \label{eq:control_bound}
   \end{equation}
   where
   \[
   \tilde{k}_\sigma = \sqrt{\lambda_{\max} (X^{1/2} K_{\sg}^\top K_{\sg} X^{1/2})}.
   \]
\item the closed-loop system \eqref{eq:switched_linear}, \eqref{eq:u} is \( \ds \)-homogeneous of degree \( \mu \) and:
\begin{itemize}
   \item globally uniformly finite-time stable for \( \mu \in [-1, 0) \), where
     \begin{equation}
     x(t) = 0, \quad \forall t \geq T(x_0) \leq t_0+\frac{\|x_0\|_{\ds}^{-\mu}}{-\mu \rho}
     \label{eq:finite_time_stability}
     \end{equation}
   \item \text{globally uniformly exponentially stable} for \( \mu = 0 \), where
     \begin{equation}
     \| x(t) \|_{\ds} \leq \E^{-\rho t} \|x_0\|_{\ds}, \quad \forall t \geq 0.
     \label{eq:exp_stability}
     \end{equation}
   \item \text{globally uniformly nearly fixed-time stable} for \( \mu > 0 \), where for any \( r > 0 \),
     \begin{equation}
     \| x(t) \|_{\ds} \leq r, \quad \forall t \geq \frac{1}{\mu \rho},
     \label{eq:nearly_fixed_time_stability}
     \end{equation}
     independently of \( x_0 \in \mathbb{R}^n \). 
     \end{itemize}
    \end{itemize} 
\end{theorem}
\begin{IEEEproof}
The proof of Theorem \ref{theorem:free_disturbances} follows the main steps of Theorem~\ref{thm:2020} in~\cite{Zimenko2020a} and Theorem~9.1 in~\cite{Polyakov2025}, which establishes the feasibility of the algebraic equation system given by \eqref{eq:homogeneous_eq0}--\eqref{eq:homogeneous_identitiesf}. Additionally, the stability properties of the switched system \eqref{eq:switched_linear} are derived using Theorem~\ref{thm:Orlov}.
\subsubsection*{1)} If the algebraic system \eqref{eq:homogeneous_eq0}--\eqref{eq:homogeneous_identitiesf} is feasible and  have a (common) solution $\{\Gh_0, Y_{0,\sg}\}$; then the matrix $I_n - \Gh_0$ is always invertible. Indeed, if we suppose the contrary then there exists a left eigenvector $h_1\in \R^{n}$ such that
\[
h_1^{\top}\Gh_0=h_1^{\top}.
\]
Since $\Gh_0$ is the solution of \eqref{eq:homogeneous_eq0}--\eqref{eq:homogeneous_eq} then  
\[
h_1^{\top}\Gh_0B_{\sigma}=h_1^{\top}B_{\sigma}=0
\]
and 
$$
h_1^{\top}A_{\sigma}\Gh_0=2h_1^{\top}A_\sg,
$$  i.e., $h_2=A_{\sigma}^{\top}h_1$ is the  left eigenvector of $\Gh_0$ too. Notice that $h_2\neq 0$ (otherwise $h_1^{\top}A_{\sigma}=0$ and $h_1^{\top}B_{\sigma}=0$ imply that the pair $\{A_{\sigma},B_{\sigma})$ is not controllable). Repeating the same considerations we derive that 
$h_i=(A_{\sigma}^{\top})^{i-1}h_1$ with $i=1,...,n$ are eigenvectors 
$
h_i^{\top}\Gh_0=ih_i
$
of $\Gh_0$ such that 
\[
h_i^{\top} B_{\sigma}=0, i=1,2,...,n.
\]
Since the eigenvectors $h_i$ with  $i=1,2,...,n$ correspond to disjoint eigenvalues $\lambda_i=1,2,\dots,n,$ then they are linearly independent, and the above identity is possible only if $B_{\sigma}=0$. However, this contradicts to the controlability of the pair $\{A_{\sigma},B_{\sigma}\}$.

The matrix $\Gh_\ds = I_n + \mu \Gh_0$ is anti-Hurwitz for $\mu$ close to zero. Moreover, if the matrix $I_n-\Gh_0$ is invertible then   $K_{0,\sg} = Y_{0,\sg} (\Gh_0 - I_n)^{-1}$ and $A_{0,\sg} = A_\sg + B_\sg K_{0,\sg}$, we have
\[
A_\sg \Gh_0 - \Gh_0 A_\sg + B_\sg K_{0,\sg} (\Gh_0 - I_n) = A_\sg, \quad \Gh_0 B_\sg = 0,
\]
which implies
\[
(A_\sg + B_\sg K_{0,\sg}) \Gh_0 = \Gh_0 A_\sg + \Gh_0 B_\sg K_{0,\sg} + A_\sg + B_\sg K_{0,\sg}.
\]
Thus,
\[
\mu A_{0,\sg} \Gh_0 = \mu (I_n + \Gh_0) A_{0,\sg},
\]
i.e., $A_{0,\sg}$ is $\ds$-homogeneous of degree $\mu$ (see Corollary \ref{cor:linearswitching}).
\subsubsection*{2)}
Controllability of the pair  $\{A_\sg, B_\sg\}$ implies the feasibility of the matrix inequalities    \eqref{eq:homogeneous_identitiesf}  for any given $ K_{0,\sigma} \in \mathbb{R}^{m \times n}$.  
 Therefore, a  solution $(X, Y_\sg )$ of the linear algebraic systems \eqref{eq:homogeneous_identitiesf}
always exists and can be constructed by means of the proposed algorithm (see, an example for fixed $\sigma$ in \cite{Polyakov2016,Polyakov2025}).

Take $1)$ and $2)$ to be verified. If the pair \((X, Y_\sigma)\) is a  solution of \eqref{eq:homogeneous_identitiesf}, then the dilation \(\ds(s) = \E^{s \Gh_\ds}\) is monotone with respect to the weighted Euclidean norm \(\|x\| = \sqrt{x^\top P x}\), where \(P = X^{-1}\). Let us show that \(\|\cdot\|_\ds\) is a Lyapunov function for the closed-loop system \eqref{eq:switched_linear} and \eqref{eq:u}. Indeed, it is globally proper, and using the following formula \cite{Polyakov2025}
\[
\frac{\partial \|x\|_\ds}{\partial x} = \|x\|_\ds  \frac{x^\top \ds^\top(-\ln \|x\|_\ds) P\ds(-\ln \|x\|_\ds)}{x^\top \ds^\top(-\ln \|x\|_\ds) P \Gh_\ds\ds(-\ln \|x\|_\ds) x}, \quad x \neq 0.
\]
Denoting $\pi_\ds(x) = \ds(-\ln \|x\|_\ds) x$ and using the equations in $1)$ and $2)$ of Theorem \ref{theorem:free_disturbances}, we have
\begin{gather*}
\frac{d}{dt} \|x\|_\ds = \frac{\partial \|x\|_\ds}{\partial x} \dot{x} 
= \|x\|_\ds  \frac{\pi_\ds(x)^\top P\ds(-\ln \|x\|_\ds)(A_\sg x+B_\sg u)}{\pi_\ds(x)^\top  P \Gh_\ds \pi_\ds(x)}\\ 
= \|x\|_\ds  \frac{\pi_\ds(x)^\top P\ds(-\ln \|x\|_\ds)(A_{0,\sg} x+\|x\|^{1+\mu}B_\sg K_\sg \ds(-\ln \|x\|_\ds) x)}{\pi_\ds(x)^\top  P \Gh_\ds \pi_\ds(x)},
\end{gather*}
Using (\ref{eq:homogeneous_identities}) and from Definition \ref{def:vector_fied} this means that
\[A_{0,\sg}\ds(s) =\E^{\mu s}\ds(s)A_{0,\sg},\] and \[\ds (s) B_\sg =\sum_{i=0}^{\infty} \frac{s^i \Gh_\ds^i B_\sg}{i!} = \E^s B_\sg \quad \text{for all} \quad s\in\R.\]
Hence, we have 
\begin{gather*}
\frac{d}{dt} \|x\|_\ds = \|x\|_\ds^{1+\mu} \frac{\pi^\top_\ds(x) P (A_{0,\sg}+ B_\sg K_\sg) \pi_\ds(x)}{\pi_\ds(x)^\top  P \Gh_\ds \pi_\ds(x)}\\
= \|x\|_\ds^{1+\mu} \frac{\pi^\top_\ds(x) P \left(A_{0,\sg}P^{-1} + B_\sg Y_\sg + P^{-1}A_{0,\sg}^\top+Y^\top_\sg B_\sg^\top\right) P \pi_\ds(x)}{\pi_\ds(x)^\top   \left(\Gh_\ds P^{-1} + P^{-1}\Gh_\ds^\top\right) P\pi_\ds(x)}
\end{gather*}
By using the inequalities in \eqref{eq:homogeneous_identitiesf}, we conclude that \begin{gather*}
\frac{d}{dt} \|x\|_\ds \leq -\rho\|x\|_\ds^{1+\mu}.
\end{gather*}
 Therefore, $\|x\|_\ds$ is a (common) Lyapunov function of the system, and the identity (\ref{eq:lyapunov_decay}) holds. For $\mu \neq 0$, the estimates      \eqref{eq:finite_time_stability},  \eqref{eq:exp_stability}, and 
     \eqref{eq:nearly_fixed_time_stability}
 follow from the identity (\ref{eq:lyapunov_decay}), which gives
\[
\|x(t)\|_d^{-\mu} \leq \|x_0\|_d^{-\mu} + \mu \rho t, \quad \forall t \geq 0 : \|x_0\|_d^{-\mu} + \mu \rho t \geq 0.
\]

Following \cite{Polyakov2025}, the canonical homogeneous norm $\|\cdot\|_\ds$ is analytic on $\mathbb{R}^n \setminus \{0\}$, then $u$ is analytic on $\mathbb{R}^n \setminus \{0\}$ as a composition of analytic functions. Moreover,
\[
\|\ds(-\ln \|x\|_\ds) x\| = \|\pi_\ds(x)\| = \sqrt{\pi_\ds^\top(x) X^{-1} \pi_\ds(x)} = 1,
\]
and for $x \neq 0$, we have
\begin{gather*}
|u(x)| \leq |K_{0,\sg} x| + \|x\|_\ds^{\mu+1} |K_\sg \pi_\ds(x)|  \\
\leq |K_{0,\sg} x| + \|x\|_\ds^{\mu+1} \sqrt{\pi_d^\top(x) K_\sg^\top K_\sg \pi_d(x)}\\
= |K_{0,\sg} x|  + \|x\|_\ds^{\mu+1} \sqrt{\pi^\top_\ds(x) X^{-\frac{1}{2}} X^{\frac{1}{2}} K_\sg^\top K_\sg X^{-\frac{1}{2}} X^{\frac{1}{2}} \pi_\ds(x)}\\
\leq  |K_{0,\sg} x|  + \|x\|_\ds^{\mu+1} \sqrt{\lambda_{\max} (X^{\frac{1}{2}} K_\sg^\top K_\sg  X^{\frac{1}{2}})}.
\end{gather*}
According to Theorem 7.5 of \cite{Polyakov2025}, any asymptotically stable homogeneous system with \( \mu < 0 \) (respectively, \( \mu > 0 \)) exhibits finite-time stability (respectively, nearly fixed-time stability), thereby completing the proof of Theorem \ref{theorem:free_disturbances}.
\end{IEEEproof}
\begin{remark}
For \( \mu = 0 \), the problem reduces to standard linear control design \cite{liberzon2003}. The homogeneous norm \( \|\cdot\|_\ds \), defined implicitly by (\ref{eq:u}), can be computed using numerical approches such as the bisection method \cite{Polyakov2016,Polyakov2025}.
Additionally, the feedback matrix \( K\sz \) can be selected using pole placement techniques.
\end{remark}

The following theorem investigates the behavior of the system in the presence of disturbances. Specifically, it extends the analysis to handle additive disturbances that affect the system dynamics, providing stability guarantees and performance bounds under these conditions.
\begin{theorem}\label{thm:disturbances}
Let all conditions of Theorem \ref{theorem:free_disturbances} be verified, and let \(\omega : \mathbb{R} \times \mathbb{R}^n \to \mathbb{R}^n \) be a locally bounded measurable function that  satisfies the following inequality
\begin{equation}\label{eq:dis_con}
\begin{aligned}
\frac{x^\top \ds^\top (- \ln \|x\|_\ds) P \ds(- \ln \|x\|_\ds) E_\sg \omega(t,x)}{ x^\top \ds^\top (- \ln \|x\|_\ds) P \Gh_\ds \ds (- \ln \|x\|_\ds) x }
\leq \kappa \|x\|_\ds^\mu, \\
\quad \forall x \in \mathbb{R}^n,
\end{aligned}
\end{equation}
and almost all $t\geq0$, for some \( \kappa \in (0, \rho) \), then for (\ref{eq:switched_linear}), (\ref{eq:u})
\begin{equation*}
    \frac{d}{dt} \|x\|_\ds \leq -\eta \|x\|_\ds ^{1+\mu}, \quad \eta = \rho - \kappa.
\end{equation*}

Consequently, the switched system (\ref{eq:switched_linear}) with the control (\ref{eq:u})  is:
\begin{itemize}
    \item \text{finite-time stable} for \( \mu < 0 \), and the solution satisfies
    \begin{equation*}
        x(t) = 0, \quad \forall t \geq t_0 + \frac{\|x_0\|_\ds^{-\mu}}{-\mu\eta}, \quad \forall x_0 \in \mathbb{R}^n.
    \end{equation*}
    \item \text{exponentially stable} for \( \mu = 0 \), with the bound
    \begin{equation*}
        \|x(t)\| \leq \E^{-\eta(t - t_0)} \|x_0\|, \quad \forall t \geq t_0.
    \end{equation*}
    \item \text{nearly fixed-time stable} for \( \mu > 0 \), satisfying
    \begin{equation*}
        \|x_{x_0}(t)\|_\ds \leq r, \quad \forall t \geq t_0 + \frac{1}{ \mu\eta}, \quad \forall r > 0, \quad \forall x_0 \in \mathbb{R}^n.
    \end{equation*}
\end{itemize}
\end{theorem}

\begin{IEEEproof}
 To analyze the time evolution of \( \|x\|_\ds \), we compute its time derivative along the trajectories of the closed-loop system (\ref{eq:switched_linear}), (\ref{eq:u}). This leads to the following expression:
\begin{gather*}
    \frac{d}{dt} \|x\|_\ds \leq -\rho \|x\|_\ds^{1+\mu}\\ + \|x\|_\ds \frac{x^\top \ds^\top (- \ln \|x\|_\ds) P \ds(- \ln \|x\|_\ds) E_{\sigma}\omega(t,x)}{ x^\top \ds^\top (- \ln \|x\|_\ds) P \Gh_\ds \ds (- \ln \|x\|_\ds) x }.
\end{gather*}
Using (\ref{eq:dis_con}) we derive the inequality:
\[
\frac{d}{dt} \|x\|_\ds \leq -\eta \|x\|_\ds^{1+\mu}, \; \text{thus, the proof is concluded}.\]
\end{IEEEproof}

Given the disturbance constraint in inequality (\ref{eq:dis_con}), the control law (\ref{eq:u}) can effectively reject  disturbances. It is also noteworthy that this result holds without requiring the disturbance to be matched. Nevertheless, the current formulation of the disturbance restriction (\ref{eq:dis_con}) is implicit.
\begin{remark}
A more simple characterization of \( E_\sg \omega \) can be obtained in the case of matched perturbations:
\(
E_\sg = B_\sg.
\)
\end{remark}
Theorem~\ref{theorem:free_disturbances} is based on the existence of a
common homogeneous Lyapunov function, which ensures stability under
arbitrary switching.
This requirement may be conservative when the subsystem dynamics differ
significantly.
To reduce conservatism, we next consider mode-dependent homogeneous
Lyapunov functions.
In this case, stability is preserved by imposing a dwell-time
condition on the switching signal.
\begin{theorem}
\label{theorem:free_disturbances_MLF}
Consider the switched linear system \eqref{eq:switched_linear} with
$\omega(t,x)=0$.
Assume that:
\begin{enumerate}
\item The linear algebraic equations
\eqref{eq:homogeneous_eq0}–\eqref{eq:homogeneous_eq}
admit a (common) solution $(\Gh_0,Y_{0,\sigma})$, and define
\[
\Gh_\ds=I_n+\mu \Gh_0,
\]
which is anti-Hurwitz for $\mu$ sufficiently close to zero. Moreover,
\[
A_{0,\sigma}=A_\sigma+B_\sigma K_{0,\sigma}, \quad
K_{0,\sigma}=Y_{0,\sigma}(\Gh_0-I_n)^{-1},
\]
satisfies
\[
A_{0,\sigma}\Gh_\ds=(\Gh_\ds+\mu I_n)A_{0,\sigma}, \quad
\Gh_\ds B_\sigma=B_\sigma.
\]
\item For each $\sigma\in\Sigma$, there exist matrices
$X_\sigma\succ0$, $Y_\sigma$, and constants $\rho_\sigma>0$ such that
\begin{align}
\label{eq:MLF_LMI}
\begin{aligned}
X_\sigma A_{0,\sigma}^\top + A_{0,\sigma}X_\sigma
+ B_\sigma Y_\sigma + Y_\sigma^\top B_\sigma^\top \\
\qquad + \rho_\sigma(\Gh_\ds X_\sigma+X_\sigma\Gh_\ds^\top) &\preceq 0,\\
\Gh_\ds X_\sigma + X_\sigma\Gh_\ds^\top &\succ 0.
\end{aligned}
\end{align}
\end{enumerate}
Define the feedback law
\begin{align}\label{eq:u_M}
u(x)=K_{0,\sigma}x
+\|x\|_{\ds,\sigma}^{1+\mu}
K_\sigma \ds(-\ln\|x\|_{\ds,\sigma})x,
\
K_\sigma=Y_\sigma X_\sigma^{-1}.
\end{align}
Then:
\begin{itemize}
\item For each mode $\sigma$, define the respective Lyapunov function $V_\sigma (x)$ as the canonical homogeneous norm $\|\cdot\|_{\ds,\sigma}$ induced by by the norm
\(\sqrt{x^\top P_\sigma x}\).
\item Along the trajectories of mode $\sigma$,
\[
\dot V_\sigma(x)\le -\rho_\sigma V_\sigma^{1+\mu}(x),
\quad \forall x\neq0.
\]
\item The closed-loop system is $\ds$-homogeneous of degree $\mu$.
\item If the switching signal satisfies a minimum dwell-time condition
\[
\inf_i (t_{i+1}-t_i) \ge \tau  > 0,
\]
then the origin of the closed-loop system is locally 
finite-time stable for $\mu\in[-1,0)$ and globally  fixed-time
stable with respect to a compact set for $\mu>0$.

\item If the switching signal satisfies an average dwell-time condition
\begin{align}\label{eq:ADT}
N_\sigma(t,t_0)\le N_0+\frac{t-t_0}{\tau_d},
\qquad
\tau_d>\frac{\ln\gamma}{\rho_{\min}},
\end{align}
then the origin of the closed-loop system is globally uniformly exponentially
stable for $\mu=0$.
\end{itemize}
\end{theorem}

\begin{IEEEproof}
Consider the mode-dependent Lyapunov functions
\(V_\sigma\). 
By item~2) of Theorem~\ref{theorem:free_disturbances_MLF}, each matrix $X_\sigma \succ 0$, hence $V_\sigma$ is positive definite and radially unbounded.

Along trajectories evolving in mode $\sigma$, using~\eqref{eq:u_M} the closed-loop system satisfies
\begin{align}\label{eq:close_loop_dewelltime}
\dot{x} = A_{0,\sigma} x + \|x\|_{\ds,\sigma}^{1+\mu} B_\sigma K_\sigma \pi_\ds(x),  
\end{align}
where $\pi_\ds(x)=\ds(-\ln \|x\|_{\ds,\sigma})x$. Using the $\ds$-homogeneity property and the LMIs of item~2), it follows that
\begin{align}\label{eq:ILF_Dynamics}
\dot V_\sigma(x) \le -\rho_\sigma V_\sigma^{1+\mu}(x), \quad \forall x\neq 0,\end{align}
which implies that each $V_\sigma$ strictly decreases along the flow of its corresponding mode.

Since any canonical norm $\|\cdot\|_\ds$ is \ds-homogeneous of degree 1 by construction, there is $\gamma>1$ such that for any $\sigma$, $\sigma'\in\Sigma$
\begin{align}\label{eq:Vssp}
V_{\sigma'}^{|\mu|}(x) \le \gamma\, V_\sigma^{|\mu|}(x), \quad \forall x\in\mathbb{R}^n, 
\end{align}
which bounds the Lyapunov function jump across modes. 

\textbf{Finite/ nearly fixed-time convergence:} The main idea for achieving convergence under the minimum dwell-time condition
\[
\inf_k (t_{k}-t_{k-1}) \ge \tau > 0
\]
is that the mode-dependent Lyapunov function $V_{\sigma_k}(t)$ strictly decreases over each dwell interval $[t_{k-1}, t_k)$.

For $\mu\in[-1,0)$, $\rho_\sigma>0$ using~\eqref{eq:ILF_Dynamics}, we have 
\[V_\sigma(t)^{-\mu} \leq V_\sigma\left(t_0\right)^{-\mu}+\mu   \rho_\sigma  \left(t-t_0\right)\]
Let the switching times be  $ 0=t_0<t_1\dots   t_{k-1}<t_{k}$ and define $\sigma(t) =\sigma(t_{k-1}) =\sigma_{k-1}$ if $t_{k-1} \leq t<t_{k}.$ 

The behavior along the trajectories using \eqref{eq:Vssp} is:
\begin{gather*}
t_0 \leq t<t_1 \quad V_{\sigma_0}\left(t_1\right)^{-\mu} \leq V_{\sigma_0}(0)^{-\mu}+\mu  \rho_{\sigma_0} \left(t_1-t_0\right)\\[1mm]
t_1 \leq t<t_2,\ V_{\sigma_1}\left(t_1\right)^{-\mu} \leq \gamma  V_{\sigma_0}\left(t_1\right)^{-\mu} \leq \gamma \bigg[V_{\sigma_0}(0)^{-\mu}+ \\[1mm]\mu \rho_\sigma \left(t_1-t_0\right)\bigg], \ V_{\sigma_1}\left(t_2\right)^{-\mu} \leq V_{\sigma_1}\left(t_1\right)^{-\mu}+\mu  \rho_{\sigma_1} \left(t_2-t_1\right)\\[1mm]  \leq \gamma  V_{\sigma_0}(0)^{-\mu}+\gamma  \mu  \rho_{\sigma_0} \left(t_1-t_0\right)+\mu  \rho_{\sigma_1} \left(t_2-t_1\right)\\[1mm]
t_2 \leq t<t_3, \ V_{\sigma_2}\left(t_2\right)^{-\mu} \leq \gamma V_{\sigma_1}\left(t_2\right)^{-\mu}, \\[1mm]  V_{\sigma_2}\left(t_3\right)^{-\mu} \leq 
V_{\sigma_2}\left(t_2\right)^{-\mu}+\mu  \rho_{\sigma_2} \left(t_3-t_2\right) \leq \gamma^2  V_{\sigma_0}(0)^{-\mu}\\+\gamma^2  \mu  \rho_{\sigma_0} \left(t_1-t_0\right)+\gamma  \mu  \rho_{\sigma_1} \left(t_2-t_1\right)+\mu  \rho_{\sigma_2} \left(t_3-t_2\right)
\\[1mm]
\cdots\\[1mm]
t_{k-1} \leq t<t_k \quad V_{\sigma_{k-1}}\left(t_k\right)^{-\mu} \leq \gamma^{k-1}  \Bigg(V_{\sigma_0}(0)^{-\mu}\\[1mm] +\mu  \sum_{i=0}^{k-1}\left[\gamma^{-i}  \rho_{\sigma_i} \left(t_{i+1}-t_i\right)\right]\Bigg).
\end{gather*}
Using the minimum dwell-time condition
\begin{align}\label{eq:mdt}
\inf_k (t_k - t_{k-1}) \ge \tau > 0
\ \text{
and defining} \ 
\rho_{\min} := \min_{\sigma \in \Sigma} \rho_\sigma,
\end{align}
we can bound the sum in the Lyapunov inequality:
\begin{gather*}
\sum_{i=0}^{k-1} \gamma^{-i} \rho_{\sigma_i} (t_{i+1} - t_i) 
\geq \sum_{i=0}^{k-1} \gamma^{-i} \rho_{\min} (t_{i+1} - t_i) 
\ge \\ \rho_{\min} \sum_{i=0}^{k-1} \gamma^{-i} \tau.
\end{gather*}
The geometric sum evaluates to
\(
\sum_{i=0}^{k-1} \gamma^{-i} \leq \frac{\gamma}{\gamma-1}.
\)

Hence, the Lyapunov function at the switching instants satisfies
\[
V_{\sigma_{k-1}}(t_k)^{-\mu} \le \gamma^{k-1}\Bigg( V_{\sigma_0}(0)^{-\mu} + \mu \, \rho_{\min} \, \tau \, \frac{\gamma}{\gamma - 1}\Bigg).
\]

If \(\mu < 0\), the second term is negative, ensuring strict decrease. Finite-time convergence is guaranteed if
\begin{align}\label{eq:cdt_FTS}
V_{\sigma_0}(0)^{-\mu} + \frac{\mu \, \rho_{\min}\, \gamma \, \tau}{\gamma - 1} \le 0
\ \Rightarrow \
\tau \ge -\frac{\gamma - 1}{\mu \, \gamma \rho_{\min}} V_{\sigma_0}(0)^{-\mu}.
\end{align}

Under this condition, the Lyapunov function reaches zero in finite time and the settling time along each mode \(\sigma\) is
\[
T_\sigma(x_0) \leq t_0+\frac{\|x_0\|_{\ds,\sigma}^{-\mu}}{-\mu \, \rho_\sigma}.
\]
Thus,
\[
\exists T < \infty : \quad V_{\sigma(t)}(t) = 0 \quad \text{and} \quad x(t) = 0, \ \forall t \ge T.
\]
Note that each mode is individually finite-time stable, there is a finite number of switching under the minimum dwell-time condition, therefore $|x|$ is always bounded. 
The above analysis establishes finite-time convergence only on compact
sublevel sets of \(V_\sigma\). 


For $\mu>0$, we have
\(V_\sigma(t)^{-\mu} \geq V_\sigma(t_0)^{-\mu} + \mu \, \rho_\sigma \, (t - t_0).\)
Using (\ref{eq:Vssp}), the behavior along the trajectories is:
\begin{gather*}
t_0 \leq t < t_1: \
V_{\sigma_0}(t_1)^{-\mu} \geq V_{\sigma_0}(0)^{-\mu} + \mu \, \rho_{\sigma_0} \, (t_1 - t_0),\\[1mm]
t_1 \leq t < t_2:\ 
V_{\sigma_1}(t_1)^{-\mu} \geq \frac{1}{\gamma} \, V_{\sigma_0}(t_1)^{-\mu} 
\geq \frac{1}{\gamma} \Big[V_{\sigma_0}(0)^{-\mu} \\[1mm] + \mu \, \rho_{\sigma_0} \, (t_1 - t_0)\Big], \ 
V_{\sigma_1}(t_2)^{-\mu} \geq V_{\sigma_1}(t_1)^{-\mu} + \mu \, \rho_{\sigma_1} \, (t_2 - t_1)\\
\geq \frac{1}{\gamma} \, V_{\sigma_0}(0)^{-\mu} + \frac{1}{\gamma} \, \mu \, \rho_{\sigma_0} \, (t_1 - t_0) + \mu \, \rho_{\sigma_1} \, (t_2 - t_1),\\[1mm]
t_2 \leq t < t_3: \
V_{\sigma_2}(t_2)^{-\mu} \geq \frac{1}{\gamma} \, V_{\sigma_1}(t_2)^{-\mu},\\[1mm]
V_{\sigma_2}(t_3)^{-\mu} \geq V_{\sigma_2}(t_2)^{-\mu} + \mu \, \rho_{\sigma_2} \, (t_3 - t_2)
\geq \gamma^{-2} \, V_{\sigma_0}(0)^{-\mu}\\[1mm] + \gamma^{-2} \, \mu \, \rho_{\sigma_0} \, (t_1 - t_0)
+ \gamma^{-1} \, \mu \, \rho_{\sigma_1} \, (t_2 - t_1) + \mu \, \rho_{\sigma_2} \, (t_3 - t_2),\\[1mm]
\dots \\[1mm]
t_{k-1} \leq t < t_k: \
V_{\sigma_{k-1}}(t_k)^{-\mu} \geq \gamma^{1-k} \, \Bigg(V_{\sigma_0}(0)^{-\mu}\\[1mm]
+ \mu \sum_{i=0}^{k-1} \gamma^{-i} \, \rho_{\sigma_i} \, (t_{i+1} - t_i)\Bigg).
\end{gather*}

 Using the minimum dwell-time condition~\eqref{eq:mdt}, substituting this estimate into the inequality obtained along the switching
trajectory yields
\[
V_{\sigma_{k-1}}(t_k)^{-\mu}
\ge
\gamma^{1-k}
\left(
V_{\sigma_0}(0)^{-\mu}
+
\mu\rho_{\min}\tau \sum_{i=0}^{k-1} \gamma^{-i}
\right).
\]

Noting that
\[
\sum_{i=0}^{k-1} \gamma^{-i}
=
\frac{\gamma - \gamma^{1-k}}{\gamma-1}, \quad \gamma>1,
\]
we obtain
\[
V_{\sigma_{k-1}}(t_k)^{-\mu}
\ge
\gamma^{1-k} \left(V_{\sigma_0}(0)^{-\mu}
+
\mu\rho_{\min}\tau \frac{\gamma - \gamma^{1-k}}{\gamma-1}\right).
\]

Since $\mu>0$, the mapping $s \mapsto s^{-1/\mu}$ is strictly decreasing.
Therefore,
\[
V_{\sigma_{k-1}}(t_k)
\le
\left(
\gamma^{1-k} V_{\sigma_0}(0)^{-\mu}
+
\mu\rho_{\min}\tau \gamma^{1-k}\frac{\gamma - \gamma^{1-k}}{\gamma-1}
\right)^{-1/\mu}.
\]
As we can conclude from this inequality, for any initial conditions and $k=1$ we obtain:
\[
V_{\sigma_0}(t_1) \leq \left(\mu\rho_\sigma\min\tau\frac{\gamma}{\gamma-1}\right)^{-1/\mu}
\]
and the system is globally fixed-time converging to this set.




For $\mu = 0$, the same steps can be applied, following the approach commonly used in~\cite{yang2014survey} and the condition~\eqref{eq:ADT}, ensures that the decay along flows dominates the accumulation of Lyapunov jumps, yielding global uniform  exponential stability of the switched closed-loop system. Finally, the control law is analytic on $\mathbb{R}^n\setminus\{0\}$ as a composition of analytic functions and is bounded away from the origin. This completes the proof.
\end{IEEEproof}
Recall that each implicit Lyapunov function is \ds-homogeneous of degree one, then there exist $0<c_1\leq c_2<+\infty$ such that
\begin{align}\label{eq:ILF_emode}
c_1 \|x\|_\ds \leq V_{\sigma}(x) \leq c_2 \|x\|_\ds,
\end{align}
for all $x\in\R^n$ and $\sigma\in\Sigma$.
 \begin{corollary}[FTS and nFxTS under state-dependent dwell-time]
\label{cor:sddt}
Consider the switched closed-loop system~\eqref{eq:close_loop_dewelltime} under the assumptions of Theorem~\ref{theorem:free_disturbances_MLF}.  

Suppose that the switching signal \(\sigma(t)\) satisfies a state-dependent dwell-time
\[
t_{k+1}-t_k \ge \tau_\sigma(x(t_k)), \quad \forall k \ge 0,
\]
where \(\tau_\sigma\) is \(\ds\)-homogeneous of degree \(-\mu\) and satisfies
\[
\tau_{\sigma}(x(t_k)) >
\begin{cases}
\left[\|x(t_k) \|_\ds\right]^{-\mu} \dfrac{c_2^{-\mu}-\gamma^{-1}c_1^{-\mu}}{-\mu}, & \text{if } \mu<0, \\[1ex]
\left[\|x(t_k) \|_\ds\right]^{-\mu} \dfrac{\gamma c_1^{-\mu}-c_2^{-\mu}}{\mu\,\rho_{\min}}, & \text{if } \mu>0 .
\end{cases}
\]
Then:
\begin{enumerate}
    \item If \(\mu<0\), the origin is globally finite-time stable (FTS).  
    \item If \(\mu>0\), the origin is globally nearly fixed-time stable (nFxTS).  
\end{enumerate}
\end{corollary}
\begin{IEEEproof}
For each mode \(\sigma\), the Lyapunov function \(V_\sigma(x)\) is positive definite, radially unbounded, and satisfies along flows~(\ref{eq:ILF_Dynamics}), as established in Theorem~\ref{theorem:free_disturbances_MLF}. Across switching instants from mode \(\sigma\) to \(\sigma'\), there exists \(\gamma>1\) such that~(\ref{eq:Vssp}) is verified.

\textbf{FTS case (\(\mu<0\)):}  As  was shown in the proof of Theorem~\ref{theorem:free_disturbances_MLF}:
\begin{gather*}
V_{\sigma_{k+1}} (t_{k+1})^{-\mu}\leq \gamma \left[V_{\sigma_{k}} (t_{k})^{-\mu} + \mu \rho_k(t_{k+1} -t_k)\right]\\
\leq \gamma\left[V_{\sigma_{k}} (t_{k})^{-\mu} + \mu \rho_{\min}(t_{k+1} -t_k)\right].
\end{gather*}
Using~\eqref{eq:ILF_emode}, we have
\begin{gather*}
\left[c_1\left(\|x(t_{k+1}) \|_\ds\right)\right]^{-\mu}\leq \gamma \left[\left[c_2\left(\|x(t_k) \|_\ds\right)\right]^{-\mu} + \mu \rho_{\min}(t_{k+1} -t_k)\right]\\
< \left[c_1\left(\|x(t_k) \|_\ds\right)\right]^{-\mu},
\end{gather*} where the last inequality is the condition whose validity guarantees the decay of Lyapunov function on each interval $[t_k,t_k+1]$, which leads to the desired condition:
\[t_{k+1} -t_k >  \dfrac{c_2^{-\mu}-\gamma^{-1}c_1^{-\mu}}{-\mu}\left[\|x(t_k) \|_\ds\right]^{-\mu}.\]
Hence, the origin is globally FTS (by \ds-homogeneity of the system, see Theorem~\ref{thm:switched-homogeneous}).

\textbf{nFxTS case (\(\mu>0\)):}  
As  was shown in the proof of Theorem~\ref{theorem:free_disturbances_MLF}:
\begin{gather*}
V_{\sigma_{k+1}} (t_{k+1})^{-\mu}\geq \frac{1}{\gamma}\left[V_{\sigma_{k}} (t_{k})^{-\mu} + \mu \rho_k(t_{k+1} -t_k)\right],\end{gather*}
due to $\mu >0$ and using~\eqref{eq:ILF_emode}, we have
\begin{gather*}
V_{\sigma_{k+1}} (t_{k+1})^{\mu}\leq \frac{\gamma}{V_{\sigma_{k}} (t_{k})^{-\mu} + \mu \rho_{\min}(t_{k+1} -t_k)}+\left[c_1\left(\|x(t_{k+1}) \|_\ds\right)\right]^{\mu}\\
\leq  \frac{\gamma}{\left[c_2\left(\|x(t_k) \|_\ds\right)\right]^{-\mu} + \mu \rho_{\min}(t_{k+1} -t_k)} \\
 < \left[c_1\left(\|x(t_k) \|_\ds\right)\right]^{\mu}
,\end{gather*}
 where again the last inequality has to be verified to guarantee the convergence, which concludes the desired condition:
\[t_{k+1} -t_k >  \dfrac{\gamma c_1^{-\mu}-c_2^{-\mu}}{\mu\,\rho_{\min}}\left[\|x(t_k) \|_\ds\right]^{-\mu}.\]
This establishes global nFxTS,  completing the proof (by \ds-homogeneity of the system established in Theorem~\ref{thm:switched-homogeneous}).
\end{IEEEproof}

It is worth to highlight that the result of this corollary corresponds to the one of Theorem~\ref{thm:switched-homogeneous}: the switched system is \ds-homogeneous, then to have a global convergence the dwell-time function has to be \ds-homogeneous of degree $-\mu$.

The following results address robustness with respect to disturbances.
\begin{theorem}
\label{thm:disturbances2}
Let all conditions of Theorem~\ref{theorem:free_disturbances_MLF} hold, and let
\(\omega : \mathbb{R}_{\ge 0} \times \mathbb{R}^n \to \mathbb{R}^n\)
be a locally bounded measurable disturbance.

Assume that, for each mode \(\sigma \in \Sigma\), there exists a constant
\(\kappa_\sigma \in (0,\rho_\sigma)\) such that
\begin{equation}\label{eq:dis_con_MLF}
\frac{
x^\top \ds^\top(-\ln\|x\|_{\ds,\sigma})
P_\sigma
\ds(-\ln\|x\|_{\ds,\sigma})
E_\sigma \omega(t,x)
}{
x^\top \ds^\top(-\ln\|x\|_{\ds,\sigma})
P_\sigma \Gh_\ds
\ds(-\ln\|x\|_{\ds,\sigma}) x
}
\le \kappa_\sigma \|x\|_{\ds,\sigma}^{\mu},
\end{equation}
for all \(x \neq 0\) and almost all \(t \ge 0\).

Then, along the trajectories of the closed-loop switched system
\((\ref{eq:switched_linear}),(\ref{eq:u})\),
\[
\frac{d}{dt} \|x\|_{\ds,\sigma}
\le
- \eta_{\sigma(t)} \|x\|_{\ds,\sigma}^{1+\mu},
\quad
\eta_\sigma := \rho_\sigma - \kappa_\sigma > 0.
\]
\item If the switching signal satisfies a minimum dwell-time condition
\begin{align}\label{eq:MDT2}
\inf_i (t_{i+1}-t_i) \ge \tau^\prime  > 0,
\end{align}
then the origin of the closed-loop system is locally
finite-time stable for $\mu\in[-1,0)$ and globally  fixed-time
stable with respect to a compact set for $\mu>0$.
\item If the switching signal satisfies an average dwell-time condition
\begin{align}\label{eq:ADT2}
N_\sigma(t,t_0)\le N_0+\frac{t-t_0}{\tau_d},
\
\tau_d>\frac{\ln\gamma}{\eta_{\min}}, \ 
\eta_{\min} := \min_{\sigma \in \Sigma} \eta_\sigma,
\end{align}
then the origin of the closed-loop system is globally uniformly exponentially
stable for $\mu=0$.
\end{theorem}
\begin{IEEEproof}
Consider the mode-dependent Lyapunov functions
\(
V_\sigma\). 
By Theorem~\ref{theorem:free_disturbances_MLF}, each \(V_\sigma\) is positive
definite and satisfies, along trajectories of mode \(\sigma\),
\begin{gather*}
\dot V_\sigma(x)
\le
-\rho_\sigma V_\sigma^{1+\mu}(x)
\\ + V_\sigma(x)
\frac{
x^\top \ds^\top(-\ln V_\sigma)
P_\sigma
\ds(-\ln V_\sigma)
E_\sigma \omega(t,x)
}{
x^\top \ds^\top(-\ln V_\sigma)
P_\sigma \Gh_\ds
\ds(-\ln V_\sigma) x
}.
\end{gather*}

Using condition \eqref{eq:dis_con_MLF}, it follows that
\[
\dot V_\sigma(x)
\le
-(\rho_\sigma - \kappa_\sigma) V_\sigma^{1+\mu}(x)
= -\eta_\sigma V_\sigma^{1+\mu}(x),
\quad \forall x \neq 0.
\]
The remaining steps are the same steps as in the proof of  Theorem~\ref{theorem:free_disturbances_MLF}. 
\end{IEEEproof}
\begin{remark}[On the computation of $\|\cdot\|_\ds$] 
The homogeneous control law~\eqref{eq:u} requires an efficient computational procedure for practical implementation. An approximation scheme for evaluating $\|\cdot\|_\ds$ is proposed in~\cite{Polyakov2025}, while issues related to its real-time numerical computation are addressed in~\cite{Polyakov2016} through a bisection-based method. This numerical approach has been demonstrated to be very effective in practice (see, e.g.,~\cite{Zimenko2020a}). 
\end{remark}

\section{Simulations}
In this section, we present numerical simulations to evaluate the performance of the proposed finite-time controller applied to the switched system~\eqref{eq:switched_linear} under matched perturbations, i.e., $E_\sigma = B_\sigma$, so that the disturbance enters the system through the same channels as the control input. We also assess the performance of the nearly fixed-time controller in the presence of both matched and mismatched perturbations. The switching signal $\sigma(t)$ is illustrated in Fig.~\ref{fig:signal1}, where $0$ corresponds to the activation of the first subsystem $(A_1, B_1)$ and $1$ corresponds to the activation of the second subsystem $(A_2, B_2)$.
\begin{figure}[h!]
    \centering
    \includegraphics[width=0.45\textwidth]{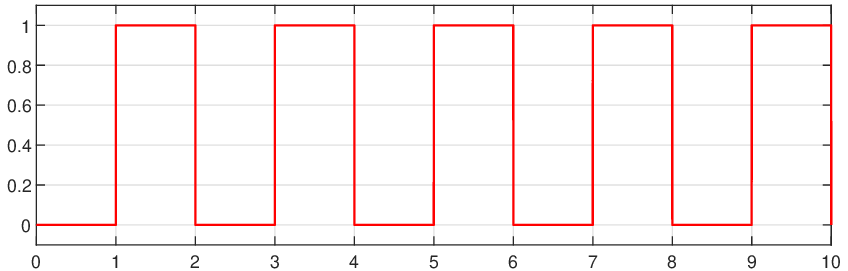} 
    \caption{Switching signal: 0 corresponds to subsystem 1 (\( A_1, B_1 \)), 1 corresponds to subsystem 2 (\( A_2, B_2 \)).}
    \label{fig:signal1}
\end{figure}

We consider the MIMO system with 
\(
\omega(t, x(t)) = 0.8 \sin(10 t), 
\)
and system matrices for \( n = 4 \):
\[
A_1 = \begin{bmatrix}
0 & 2 & 2 & 0 \\
0 & 0 & 3 & 0 \\
0 & 0 & 0 & 1 \\
0 & 0 & 0 & 0
\end{bmatrix}, \
B_1 =  \begin{bmatrix}
0 & 0 \\
0 & 0 \\
0 & 0 \\
1 & 2
\end{bmatrix}, \] \[
A_2 = \begin{bmatrix}
0 & 1 & 0 & 0 \\
0 & 0 & 3 & 0 \\
0 & 0 & 0 & 1 \\
0 & 0 & 0 & 0
\end{bmatrix}, \
B_2 = \begin{bmatrix}
0 & 0 \\
0 & 0 \\
0 & 0 \\
-1 & -2
\end{bmatrix}.
\]
The finite-time control law \( u \) is defined by \eqref{eq:u} with parameters \( \mu_i = -0.1 \) and \( \rho_i = 2 \). The corresponding matrices for the controller are:
\begin{gather*} 
K_1 = \begin{bmatrix}
-40.7387 & -14.1629 & -30.3225 & -3.7321 \\
-81.4773 & -28.3257 & -60.6449 & -7.4643
\end{bmatrix}, \\
P_1 = \begin{bmatrix}
81.8695 & 7.8511 & 31.8486 & 3.0001 \\
7.8511 & 2.9934 & 5.1863 & 0.5323 \\
31.8486 & 5.1863 & 14.8742 & 1.4720 \\
3.0001 & 0.5323 & 1.4720 & 0.1753
\end{bmatrix},\\ \Gh_\ds = \begin{bmatrix}
0.6000 & 0 & 0 & 0 \\
0 & 0.7000 & 0 & 0 \\
0 & 0 & 0.8000 & 0 \\
0 & 0 & 0 & 0.9000
\end{bmatrix},\\ 
\Gh_0 = \begin{bmatrix}
0.6 & 0 & 0 & 0 \\
0 & 0.7 & 0 & 0 \\
0 & 0 & 0.8 & 0 \\
0 & 0 & 0 & 0.9
\end{bmatrix},\\ 
K_2 = \begin{bmatrix}
67.7426 & 37.6748 & 27.7519 & 3.7321 \\
135.4852 & 75.3497 & 55.5037 & 7.4643
\end{bmatrix}, \\
P_2 = \begin{bmatrix}
235.3985 & 103.6428 & 55.1037 & 5.1875 \\
103.6428 & 47.1603 & 26.0054 & 2.5019 \\
55.1037 & 26.0054 & 15.1619 & 1.5323 \\
5.1875 & 2.5019 & 1.5323 & 0.1825
\end{bmatrix}.
\end{gather*}
\begin{figure}[htbp]
    \centering
    \includegraphics[width=0.5\textwidth]{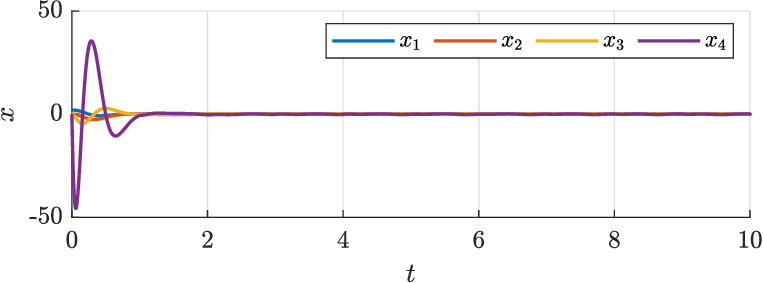} 
    \includegraphics[width=0.5\textwidth]{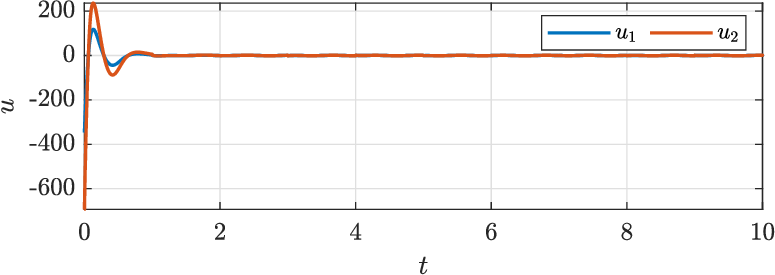} 
    \caption{System states and control input versus time for the finite-time controller.}
    \label{fig:Fig3}
\end{figure}
The simulation is performed with initial condition \( x_0 = [2, 0, 0, 0]^\top \). The time evolution of the states and control input is shown in Fig.~\ref{fig:Fig3}, demonstrating the effectiveness of the proposed controller in achieving finite-time stabilization.

We consider the MIMO system with 
\(
\omega(t, x(t)) = \left[0.5 \cos(2 t) 0.4 \sin(5 t)  0.4 \sin(5 t) 0.3 \cos(3 t) \right]^\top
\), $E_1 = E_2 = I_4$
and system matrices for \( n = 4 \):
\[
A_1 = \begin{bmatrix}
0 & 2 & 2 & 0 \\
0 & 0 & 3 & 0 \\
0 & 2 & 0 & 1 \\
0 & 1 & 0 & 0
\end{bmatrix}, \
B_1 =  \begin{bmatrix}
0 & 0 \\
0 & 0 \\
0 & 0.1 \\
1 & 0
\end{bmatrix}, \] \[
A_2 = \begin{bmatrix}
0 & 1 & 0 & 0 \\
0 & 0 & 3 & 0 \\
0 & 4 & 0 & 1 \\
0 & 1 & 0 & 0
\end{bmatrix}, \
B_2 = \begin{bmatrix}
0 & 0 \\
0 & 0 \\
0 & 2 \\
1 & 0
\end{bmatrix}.
\]
The finite-time control law \( u \) is defined by \eqref{eq:u_M} with parameters \( \mu_i = 1 \) and \( \rho_i = 1 \). The corresponding matrices for the controller are:
\begin{gather*} 
K_1 = \E^3\begin{bmatrix}
0     &    0   &      0 &  -0.0040\\
   -2.0744 &  -0.2667  & -0.2400     &    0
\end{bmatrix}, \\
P_1 = \E^{3}\begin{bmatrix}
3.9350   & 0.7620  &  0.1448    &     0\\
    0.7620  &  0.1629  &  0.0349     &    0\\
    0.1448   & 0.0349  &  0.0105    &     0\\
         0    &     0     &    0   & 0.0021
\end{bmatrix},\\ 
\Gh_\ds = \begin{bmatrix}
4 & 0 & 0 & 0 \\
0 & 3 & 0 & 0 \\
0 & 0 & 2 & 0 \\
0 & 0 & 0 & 2
\end{bmatrix},\quad \Gh_0 = \begin{bmatrix}
0.6 & 0 & 0 & 0 \\
0 & 0.7 & 0 & 0 \\
0 & 0 & 0.8 & 0 \\
0 & 0 & 0 & 0.9
\end{bmatrix},\\  
K_2 = \E^{3}\begin{bmatrix}
        0   &      0   &      0  & -4\\
 -163.0303 & -40.5051&  -12   &      0
\end{bmatrix}, \\
P_2 =\E^{3} \begin{bmatrix}
3.9350 &   0.7620 &   0.1448   &      0\\
    0.7620   & 0.1629   & 0.0349 &        0\\
    0.1448 &   0.0349  &  0.0105    &     0\\
         0    &     0       &  0  &  0.0021
\end{bmatrix}.
\end{gather*}
\begin{figure}[htbp]
    \centering
    \includegraphics[width=0.5\textwidth]{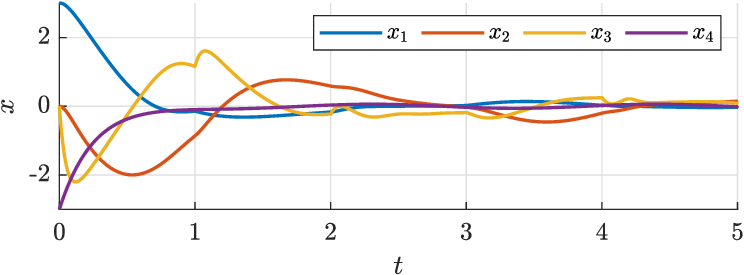} 
    \includegraphics[width=0.5\textwidth]{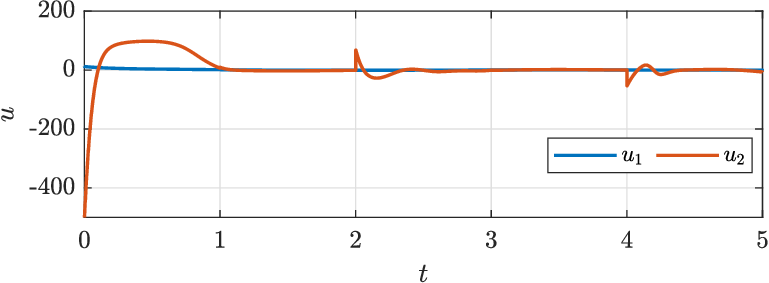} 
    \caption{System states and control input versus time for the nearly fixed-time controller.}
    \label{fig:Fig4}
\end{figure}
The simulation is performed with initial condition \( x_0 = [-3, 0, 0, 3]^\top \). The time evolution of the states and control input is shown in Fig.~\ref{fig:Fig4}, demonstrating the effectiveness of the proposed controller in achieving nearly fixed-time stabilization.

\section{Conclusions}
In this paper, homogenizing and stabilizing control laws were developed for switched linear MIMO systems. It was shown that the closed-loop system can be rendered homogeneous of a prescribed degree, leading to global uniform finite-time, exponential, or nearly fixed-time stability under a minimum or average dwell-time condition. The proposed controller synthesis is formulated in terms of linear matrix equations and inequalities, enabling a systematic design procedure. Robustness with respect to disturbances and uncertainties was analyzed and sufficient conditions ensuring preservation of the stability properties were established. Numerical simulations confirmed the effectiveness of the proposed approach.

\bibliographystyle{IEEEtran}
\bibliography{bib}

\begin{IEEEbiography}[
]{Moussa Labbadi}
(Member, IEEE) received the B.S. degree in Mechatronics Engineering from UH1 in 2015, the M.S. degree in Mechatronics Engineering from UAE in 2017, and the Ph.D. degree in Automatic Control from EMI, UM5, Rabat, Morocco, in September 2020. From 2020 to 2021, he served as a Researcher at the MAScIR Foundation, Morocco. Subsequently, from 2021 to 2022, he held the position of ATER with the Department of Automatic Control, INSA, LAMIH, and UPHF. Between 2022 and 2023, he worked as a Postdoctoral Researcher at INPG, GIPSA-Lab, and UGA. He was an associate professor at Aix-Marseille University (AMU) and a member of the Laboratory of Informatics and Systems (LIS) From 2023 to 2026. He is currently a Professor (Junior) at INP Bretagne and IRDL laboratory. 
His research interests span variable structure control, sliding mode control, observation, nonlinear system stability, control theory, and their applications.
Pr. Labbadi is an active member of several IEEE and IFAC Technical Committees and he is Associate editor for CCTA and IFAC WC 2026 and is an associate editor of IEEE T-ASE.
\end{IEEEbiography}
\begin{IEEEbiography}[
]
{Andrey Polyakov}
 received the Ph.D. degree in systems analysis and control from Voronezh State University, Voronezh, Russia in 2005.,Till 2010, he was an Associate Professor with Voronezh State University. In 2007 and 2008, he was working with CINVESTAV, Mexico City, Mexico. From 2010 to 2013, he was a Leader Researcher with the Institute of the Control Sciences, Russian Academy of Sciences, Moscow, Russia. In 2013, he joined Inria, Lille, France. He has coauthored more than 100 papers in peer-reviewed journals as well as several  books such as \emph{Attractive Ellipsoids in Robust Control},  \emph{Road Map for Sliding Mode Control Design},  \emph{Generalized Homogeneity in Systems and Control}. His research interests include various aspects of nonlinear control and estimation theory, in particular, generalized homogeneity, finite/fixed-time stability, input-to-state stability, and Lyapunov methods for both finite dimensional and infinite dimensional systems
\end{IEEEbiography}
\begin{IEEEbiography}[
]{Denis Efimov}
(Fellow, IEEE) received the Ph.D. degree in Automatic Control from Saint Petersburg State Electrotechnical University, Saint Petersburg, Russia, in 2001, and the Dr.Sc. degree from the Institute for Problems in Mechanical Engineering, Russian Academy of Sciences, Saint Petersburg, Russia, in 2006. 
From 2006 to 2011, he held research positions at L2S (CNRS–Supélec), France; the Montefiore Institute, University of Liège, Liège, Belgium; and the IMS Laboratory (CNRS), University of Bordeaux, Bordeaux, France. In 2011, he joined Inria, France, where he has been the Scientific Head of the VALSÉ Team since 2018. 
He has coauthored more than 200 scientific articles and coordinated several research projects. Dr.~Efimov is a member of several IFAC Technical Committees. He served as an Associate Editor of the \emph{IEEE Transactions on Automatic Control} and is currently an Associate Editor of \emph{Automatica} and a Senior Editor of \emph{Systems \& Control Letters}.
\end{IEEEbiography}



    


%
%
%
%
%
%

\end{document}